\documentclass[a4paper,11pt]{article}
\usepackage[sort]{cite}
\usepackage{bookmark}
\usepackage[utf8]{inputenc} 
\usepackage{bussproofs}
\usepackage{enumerate}
\usepackage{enumitem} 
\usepackage[T1]{fontenc}    
\usepackage{url}            
\usepackage{booktabs}       
\usepackage{amsfonts}       
\usepackage{amsmath}
\usepackage{amsthm}
\usepackage{amssymb}
\usepackage{mathrsfs}
\usepackage{graphicx} 
\usepackage{tabularx}
\usepackage{tikz-cd}
\usepackage{tikz}
\usepackage{proof}
\usepackage{caption}
\usetikzlibrary{matrix}
\usepackage{nicefrac}     
\usepackage{multirow} 
\usepackage{float}
\usepackage{microtype}      
\usepackage{doi}
\usepackage[nottoc]{tocbibind}
\usepackage{ragged2e}
\usepackage{float}
\usepackage{chngcntr}
\usepackage{enumerate}
\usepackage{fullpage}
\usepackage{blindtext}
\usepackage{stmaryrd }
\usepackage{nomencl}
\usepackage{etoolbox}
\usepackage[draft]{fixme}

\usepackage{xcolor}

\usepackage{enumerate}

\usepackage{amsthm}
\usepackage{amsfonts}
\usepackage{amssymb}
\usepackage{mathabx}
\usepackage{bussproofs}
\usepackage{lineno}

\usepackage{stmaryrd}
\usepackage{euscript} 
\usepackage{latexsym}
\usepackage{mathrsfs}
\usepackage[all]{xy}
\usepackage{dsfont}
\usepackage{proof}

\usepackage{graphicx}
\usepackage{relsize}

\usepackage{url}
\usepackage{enumitem}
\usepackage{xcolor}
\usepackage{colortbl}
\usepackage{colonequals}
\usepackage[all]{xy}
\usepackage{tikz}
\usetikzlibrary{positioning,arrows,calc}

\usepackage{doc}

\usepackage{url}            
\usepackage{booktabs}       
\usepackage{amsfonts}       
\usepackage{amsmath}
\usepackage{amsthm}
\usepackage{amssymb}
\usepackage{mathrsfs}
\usepackage{graphicx} 
\usepackage{tabularx}
\usepackage{tikz-cd}
\usepackage{tikz}
\usepackage{proof}
\usepackage{caption}
\usetikzlibrary{matrix}
\usepackage{nicefrac}     
\usepackage{multirow} 
\usepackage{microtype}      
\usepackage[nottoc]{tocbibind}
\usepackage{stmaryrd }
\usepackage{nomencl}
\usepackage{etoolbox}
\usepackage[draft]{fixme}
\usepackage{graphics} 
\usepackage{mathrsfs}
\usepackage{xcolor}
\usepackage{booktabs,array}
\usepackage{amssymb}
\usepackage{amsmath}
\usepackage{amsthm}
\usepackage{amsfonts}
\usepackage{amssymb}
\usepackage{tikz}

\usepackage{graphicx} 
\usepackage{booktabs} 
\usepackage{bussproofs}
\usepackage{mathabx}

\graphicspath{ {./images/} }

\usepackage{bussproofs}
\usepackage{amssymb} 
\usepackage{amsmath} 
\usepackage{stmaryrd}
\usepackage{euscript} 
\usepackage{latexsym}
\usepackage{mathrsfs}
\usepackage[all]{xy}
\usepackage{dsfont}
\usepackage{proof}

\newtheorem{theorem}{Theorem}[]
\newtheorem{corollary}{Corollary}[]
\newtheorem{lemma}{Lemma}[]
\newtheorem{proposition}{Proposition}[]
\newtheorem{remark}{Remark} []
\newtheorem{definition}{Definition}[]
\newtheorem {example}{Example}[]

\newcommand\imm{\ensuremath}
\newcommand\limp{\to}
\newcommand\impl{\limp}

\newcommand\qe[1]{\exists #1}

\newcommand\Proves\Rightarrow
\newcommand\nProves\nRightarrow

\newcommand\I{\imm{\mathcal I}}

\newcommand\nmodels{\mathrel{\nvDash}}

\newcommand\entails{\mathrel{\models}}
\newcommand\nentails{\mathrel{\nmodels}}

\usepackage{bookmark}
\usepackage[utf8]{inputenc} 
\usepackage[T1]{fontenc}    
\usepackage{url}            
\usepackage{booktabs}       
\usepackage{amsfonts}       
\usepackage{amsmath}
\usepackage{amsthm}
\usepackage{amssymb}
\usepackage{mathrsfs}
\usepackage{graphicx} 
\usepackage{tabularx}
\usepackage{tikz-cd}
\usepackage{tikz}
\usepackage{proof}
\usepackage{caption}
\usetikzlibrary{matrix}
\usepackage{nicefrac}     

\usepackage[nottoc]{tocbibind}

\newcommand{\shorttitle}{\@title}
\renewenvironment{abstract}
{
  \centerline
  {\large \bfseries \scshape Abstract}
  \begin{quote}
}
{
  \end{quote}
}

\def\keywordname{{\bfseries \emph{Keywords}}}%
\def\keywords#1{\par\addvspace\medskipamount{\rightskip=0pt plus1cm
\def\and{\ifhmode\unskip\nobreak\fi\ $\cdot$
}\noindent\keywordname\enspace\ignorespaces#1\par}}

\title{Skolemization and Decidability of the Bernays–Sch\"onfinkel Class in G\"odel Logics
\thanks{This research was funded in part by the Austrian Science Fund (FWF) 10.55776/P36571 and 10.55776/RIC1287724.}
}

\author{
Matthias Baaz, Mariami Gamsakhurdia, Anela Lolic\\
Technische Universität Wien\\
A–1040 Vienna, Austria\\
\textit{Baaz@logic.at}\\
  \textit{mariami@logic.at}\\
  \textit{Anela@logic.at}
 }

\date{}

\begin{document}

\maketitle
\begin{abstract}
In 1928, P. Bernays and M. Sch\"onfinkel proved the decidability for the class of function-free sentences with prefixes $\exists\bar{x}\forall\bar{y}A(\bar{x},\bar{y})$ (satisfiability) and $\forall \bar{x} \exists \bar{y}A(\bar{x},\bar{y})$ (validity) (specifically, the set of sentences that, when written in prenex normal form, have a prefix containing quantifiers and the matrix without
function symbols).
We will study the decidability of the Bernays–Sch\"onfinkel class for all G\"odel logics. Our argument for validity is based on the fact that Skolemization is possible for prenex G\"odel logics and our argument for satisfiability is based on the general properties of prenex formulas In G\"odel logics. (We call satisfiability 1-satisfiability, as validity and satisfiability are not dual as in classic logic and there are more concepts of satisfiability in G\"odel logics. )
Our central result shows that the validity and 1-satisfiability problems for the BS class are decidable across all Gödel logics. 
As a corollary, we obtain that valid formulas in the Bs class are valid in all infinite goedel logics and 1 satisfiable formulas are 1 satisfiable in all goedel logics.

\end{abstract}


\section{Introduction}
It is widely acknowledged that any first-order formula in classical logic is logically identical to one in prenex form. 
In general, any set of quantifier prefixes defines a fragment of first-order logic, specifically the set of prenex formulas that contain one of the quantifier prefixes in question.
In the early stages of research, it was recognised that while some fragments defined in this way have decidable satisfiability/validity, others do not.

In 1928, P. Bernays and M. Sch\"onfinkel proved the decidability for the class of function-free sentences with prefixes $\exists\bar{x}\forall\bar{y}A(\bar{x},\bar{y})$ (satisfiability) and $\forall \bar{x} \exists \bar{y}A(\bar{x},\bar{y})$ (validity) (specifically, the set of sentences that, when written in prenex normal form, have a prefix containing quantifiers and the matrix without function symbols denoted as $BS$ class) \cite{DrebenGoldfarb}.

The $BS$ class is known to be decidable in classical logic, and also in intuitionistic logic. In classical logic, the existence of Skolem functions ensures that for a formula of the form $\forall x \exists y A(x, y)$, a witnessing term $t(x)$ can be constructed. Validity reduces to checking ground instances, and this process is supported by Herbrand Theorem and the completeness of the classical system. In intuitionistic logic, decidability of the $BS$ class follows from the constructive nature of intuitionistic reasoning: existential quantifiers are interpreted via realizability, and witnessing terms must be explicitly constructible. However, this stability does not extend to general intermediate logics. Our goal is to study decidability classes in all Gödel logics.

Gödel logics form an important class of intermediate logics, situated between classical and intuitionistic logic, parameterized by closed subsets \( V \subseteq [0,1] \) containing 0 and 1. Unlike classical logic, validity and satisfiability are not dual in Gödel logics. 
The lack of duality between validity and satisfiability means that classical techniques cannot be directly applied. Therefore validity 
and satisfiability should be treated separately.
One crucial distinction is the lack of proof-theoretic support for many Gödel logics. Unlike classical and some intuitionistic systems, Gödel logics often lack recursive enumerability (r.e.) or might not admit analytic calculi (e.g., cut-free hypersequent systems). Finite-valued Gödel logics are known to support such calculi.

However, not all prenex fragments of Gödel logics are r.e.. 
A prominent example is the logic $G_\uparrow$, which is known to be not r.e.. Importantly, its prenex fragment is logically equivalent to the full logic, and therefore also not r.e..

Thus, while $BS$ formulas are decidable in all Gödel logics, this result is not a consequence of syntactic simplicity or proof-theoretic regularity.
The decidability of the $BS$ class here crucially relies on a semantic \emph{Skolemization}, that replaces universal quantifiers with constants w.r.t validity and existential quantifiers w.r.t. 1-satisfiability. This reduction yields purely existential, constant-only formulas whose validity can be decided propositionally via Herbrand-style disjunctions. The Skolemization of prenex fragments for 1-satisfiability is immediate as it coincides with classical Skolemization.




Our argument for satisfiability is based on the \emph{gluing argument}, a crucial tool for analyzing the decidability of the satisfiability of the prenex subclasses in Gödel logic and the general properties of prenex formulas in G\"odel logics.
 Specifically, we consider a redefinition of satisfiability: if a set of formulas is satisfiable in a Gödel logic with an isolated truth value from above, one can modify the interpretation such that all formulas evaluating above this isolated value are assigned the truth value $1$, while preserving the values of the remaining formulas.

\section{Preliminaries}

\begin{definition}\label{goedel}(G\"odel logics).
First-order G\"odel logics are a family of many-valued logics where the truth values set (known also as \emph{G\"odel set}) \(V\) is closed subset of the full \([0,1]\) interval that includes both $0$ and $1$ given by the following evaluation function \(\mathcal{I}\) on $V$
\begin{align*}
   (1) & \quad \mathcal{I}(\bot)= 0\\
   (2) & \quad \mathcal{I}(A \wedge B)= min \{\mathcal{I}(A), \mathcal{I}(B)\}\\
   (3) & \quad \mathcal{I}(A \vee B)= max \{\mathcal{I}(A), \mathcal{I}(B)\} \\
   (4) & \quad  \mathcal{I}(A \supset B) = \begin{cases}\I(B) & \text{if\/ $\I(A) > \I(B)$,} \\
               1     & \text{if\/ $\I(A) \le \I(B)$.}\end{cases}   \\       
 (5) & \quad \mathcal{I}(\forall x A(x))= \text{inf}\{\mathcal{I}( A(u))\: u\in U_\mathcal{I}\} \\
 (6) & \quad\mathcal{I}(\exists x A(x))= \text{sup} \{\mathcal{I}( A(u))\: u\in U_\mathcal{I}\} 
\end{align*}
\end{definition}
\begin{definition} For a truth value set $V$, a (possibly infinite) set $\Gamma$ of formulas ($1$-)entails a formula $A$ if the interpretation  $\I$ on $V$ of $A$ is $1$ in case the interpretations of all formulas in $\Gamma$ are $1$, i.e.,
$\Gamma \Vdash_V A
\Longleftrightarrow  (\forall \I,  \forall B \in \Gamma:\I(B) = 1) \rightarrow \I(A) = 1$.
\end{definition}
As a generalization of classical satisfiability, we introduce the
following concepts:
\begin{definition}[Validity]\label{d - 1.4}
The formula in Gödel logic is \emph{valid} 
if the formula evaluates to 1 
under every interpretation. 
\end{definition}
\begin{definition}[satisfiability]\label{d - 1.3} The formula in Gödel logic is \emph{1-satisfiable 
} if there exists at least one interpretation
that assigns 1 
to the formula.  
\end{definition}

\begin{remark}
    Validity and unsatisfiability are not dual in Gödel logic, e.g. $A\vee \neg A$ is not valid but its negation is unsatisfiable.
\end{remark}

Properties of first-order Gödel logics are based on the order-theoretic structure of their set of truth values rather than the truth values themselves. 
In particular, different logics are distinguished by the number and type of accumulation points and their respective Cantor–Bendixon ranks \cite{Baaz:Delta},\cite{Beckmann-Goldstern-Preining:Fraisse},\cite{Kechris:DescriptiveSetTheory}. This raises the question of an intentional\footnote{ We use intentional in the sense that logical properties (e.g., set of validities, entailment relation, set of satisfiable formulas) change when the characterising object changes.} characterization of Gödel logics. While it is technically convenient to represent Gödel logics via their set of truth values, it has the disadvantage that uncountably many truth value sets have indistinguishable logical properties. 
For example, consider 
\begin{center}
 \(V_1 = \{1/n : n\in N^+\} \cup \{0\}\) \\
\(V_2 = \{1/2n : n\in N^+\}\cup\{0\}\).   
\end{center} The logics \(G_{V_1}\) and \(G_{V_2}\) have exactly the same properties.
This leads to the representation of Gödel logics as a set of valid sentences.
\begin{definition} \label{d - 1.2} The Gödel logic \(G_V\)  is defined as the set of formulas \(A\) in the language \(\mathcal{L}_G\) such that the evaluation of \(A\) is 1 for all 
\(V\)-interpretations \(\mathcal{I}\).
 \end{definition}


We first consider several ‘prototypical’\footnote{This is because logic is defined extensionally as the set of formulas valid in this truth value set, so the Gödel logics on different truth value sets may coincide.} Gödel sets:\\
\begin{center}
\(V_{[0,1]}=[0,1]\)\\
\(V_\downarrow=\{0\}\cup \{1/k : k\geq 1\}\)\\
\(V_\uparrow=\{1\}\cup \{1 - 1/k : k\geq 1\}\)\\
 \(V_m=\{1\}\cup \{1 - 1/k : 1\geq k\geq m-1\}\)  
\end{center}
The corresponding Gödel logics are \(G_{[0,1]}\), \(G_\downarrow\), \(G_\uparrow \), \(G_m\). \(G_{[0,1]}\) is \emph{standard} Gödel logic.
The logic \(G_\downarrow\) also turns out to be closely related to some temporal logics. 

The relationships between finite and infinite valued propositional Gödel logics are well understood. Any choice of an infinite set of truth values results in the same propositional Gödel logic. In the first-order case, the situation is more interesting.
We establish several important results regarding the relationships between various first-order
Gödel logics.

\begin{proposition}\label{p - 1.1}
If $V\subseteq V'$ then $G_{V'} \subseteq G_V$.
\end{proposition}
\begin{proof}
The evaluation of a formula $A$  under an interpretation $\mathcal{I}$ depends only on the relative ordering of the truth of atomic formulas, rather than on the specific choice of the set $V'$ or the numerical values assigned to atomic formulas.

If $V \subseteq V'$ are both Gödel sets, and $\mathcal{I}$ is an interpretation into $V'$, then $\mathcal{I}$ can also be seen as an interpretation into
$V$, and the values $\mathcal{I}(A)$, computed recursively using (1) – (6) in Definition 1, do not depend on whether we consider $\mathcal{I}$ as a $V'$-interpretation or a $V$-interpretation. Consequently, if $V' \subseteq V$, there are more interpretations into $V$ than into $V'$, hence if $\Gamma\models_{V'} A$ then also $\Gamma\models_V A$ and $G_{V'}\subseteq G_V $.
\end{proof}

\begin{proposition}\label{p - 1.2}\cite{Baaz-Leitsch-Zach:FirstOrderGoedel}
    The following containment relationships hold:\\
1) \(G_m\supsetneq G_{m+1}\)\\  
2) \(G_m\supsetneq G_\uparrow \supsetneq G_{[0,1]}\)   \\
3) \(G_m\supsetneq G_\downarrow \supsetneq G_{[0,1]}\)\\
4) \(G_{[0,1]}=\bigcap_V G_V\) 
 where \(V\) ranges over all Gödel sets. \(G_{[0,1]}\) logic is contained in every Gödel logic.
\end{proposition}

\begin{proof}
The only challenging part is demonstrating strict containment.  Note that the $(Fin)$ $(\top\supset A_1)\vee(A_1\supset A_2)\vee (A_2\supset A_3)\vee \dots \vee (A_{m-2}\supset A_{m-1})\vee (A_{m-1}\supset \bot)$ is valid in 
\(G_m\) but not in \(G_{m+1}\). Consider the formulas 
\begin{center}
    (1) $\quad $ \(C_\uparrow= \exists x (A(x)\supset \forall y A(y))\)\\
   (2) $\quad $ \(C_\downarrow= \exists x (\exists y A(y)\supset  A(x))\).
\end{center}
 \(C_\downarrow\)
 is valid in all \(G_m\), 
\(G_\uparrow\). 
On the other hand,
\(C_\uparrow\) is valid in all \(G_m\) and 
\(G_\uparrow\) but not in \(G_\downarrow\). However, none of them is valid in \(G_{[0,1]}\) [\cite{Baaz-Leitsch-Zach:FirstOrderGoedel}, Corollary
2.9]. \\
If $\Gamma \entails_V A$ for every G\"{o}del set $V$\!, then it does so in
  particular for $V = [0, 1]$.  Conversely, if $\Gamma \nentails_V A$ for a
  G\"{o}del set~$V$\!, there is a $V$\!-interpretation $\I$ with $\I(\Gamma) >
  \I(A)$. Since $\I$ is also a $[0,1]$-interpretation, $\Gamma \nentails_{[0,1]}
  A$. 
\end{proof}
\begin{remark}\label{r - 1.2}
 \(C_\uparrow\) expresses the fact that every infimum 
   is a minimum, while \(C_\downarrow\) states that every supremum 
   is a maximum.
\end{remark}

\subsection{Recursive Enumerability and Limitations}

\begin{definition}[Limit point, perfect space, perfect set]\label{d - 1.6}
A \emph{non-isolated point} of a topological space is a point \(x\) such that for every open neighbourhood \(U\) of \(x\) there exists a point \(y\in U\) with \(y\neq x\). A \emph{limit point} of a topological space is a point \(x\) that is not isolated.  A space is \emph{perfect} if all its points are limit points.
A set \(P\subseteq R\) is \emph{perfect} if it is closed and together with the topology induced from \(R\) is a perfect space. 
\end{definition}
Just as Trakhtenbrot’s Theorem \cite{trakh} shows that first-order logic over finite models is not
r.e., a similar result can be obtained in Gödel logics: where $0$ is not isolated or not in the prefect set the logic is not r.e.. The topological structure of the truth set determines the logic axiomatizability and effective decidability.

 The validity in first-order G\"odel logic is characterized by the following theorem:
\begin{theorem}\cite{Baaz-Preining:Gödel–Dummett} 
A first-order G\"odel logic \(G_V\) is r.e. iff one of the following conditions is satisfied:\\
1. \(V\) is finite,\\
2. \(V\) is uncountable 
and $0$ is an isolated point,\\
3. \(V\) is uncountable, and every neighbourhood of $0$ is in the prefect subset. 
\end{theorem}

In all the remaining cases, i.e., logics with countable truth value set and those with uncountable truth value set where
there is a countable neighbourhood of $0$ and $0$ is not isolated; the respective
logics are not r.e.

\begin{remark}
    Only $G_\uparrow$ and finite valued G\"odel logics admit all quantifier shift rules and as a consequence they allow the equivalent prenex normal form. 
\end{remark}

\section{Toward Satisfiability and the Gluing Lemma}\label{2}

In this section, we introduce a crucial tool for analyzing the decidability of the satisfiability of the prenex subclasses in Gödel logic. 
We cut the truth set at the point $\omega$, we keep the lower part and we glue the upper part to 1.
In this way, we simply change the interpretations on atoms.
The ones below  \(\omega\) remain and those above $\omega$ will become $1$. In principle, it does not matter what value \(\omega\) we fix, but if $\omega$ is not isolated from the truth values set the construction will not work for the universal quantifier.
\begin{lemma}[Gluing lemma for formulas with \(\exists\)]\label{l - 2.1}
Let \(\mathcal{I}\) be an interpretation into \(V\subseteq[0,1]\). Let us fix a value \(\omega\in[0,1]\) and define 
\begin{center}
$ \mathcal{I}'(\mathcal{P} ) = \begin{cases}\mathcal{I}(\mathcal{P}) & \text{if\/ $\mathcal{I}(\mathcal{P} )\leq \omega$,} \\
               1     & \text{otherwise}\end{cases} $
\end{center}
for an atomic formula \(\mathcal{P}\) in \(\mathcal{L}^\mathcal{I}\). Then \(\mathcal{I}_\omega\) 
is an interpretation into \(V\) such that for any formula $\mathcal{B}$
\begin{center} 
 $ \mathcal{I}'(\mathcal{B} ) = \begin{cases}\mathcal{I}(\mathcal{B}) & \text{if\/ $\mathcal{I}(\mathcal{B} )\leq \omega$,} \\
               1     & \text{otherwise.}\end{cases} $
\end{center}
\end{lemma} 
\begin{proof}
The proof is by induction on the complexity of the formulas that any subformula \(\mathcal{B}\) of \(\mathcal{A}\) with $\mathcal{I}(\mathcal{B} )\leq \omega$ has $\mathcal{I}'(\mathcal{B} )=\mathcal{I}(\mathcal{B} )$. For \(\mathcal{B}\) being an atomic formula it is obvious by definition. 
We distinguish cases according to the logical form of \(\mathcal{B}\):\\
$\mathcal{B} \equiv D \land E$.  If $\I(\mathcal{B}) \le \omega$, then, without loss of
    generality, assume
    $\I(\mathcal{B}) = \I(D)\le \I(E)$.  By induction hypothesis,
    $\I'(D) = \I(D)$ and $\I'(E) \ge \I(E)$, so $\I'(\mathcal{B}) = \I(\mathcal{B})$. If
    $\I(\mathcal{B})>\omega$, then $\I(D)>\omega$ and $\I(E)>\omega$, by induction
    hypothesis $\I'(D)=\I'(E)=1$, thus, $\I'(\mathcal{B})=1$.\\
    $\mathcal{B} \equiv D \lor E$.  If $\I(\mathcal{B}) \le v$, then, without loss of
    generality, assume
    $\I(\mathcal{B}) = \I(D) \ge \I(E)$.  By induction hypothesis,
    $\I'(D) = \I(D)$ and $\I'(E) = \I(E)$, so $\I'(\mathcal{B}) = \I(\mathcal{B})$.
    If $\I(\mathcal{B})>\omega$, then, again without loss of generality, $\I(\mathcal{B})=\I(D)>\omega$, by induction hypothesis $\I'(D)=1$, thus,
    $\I'(\mathcal{B})=1$.\\
    $\mathcal{B} \equiv D \impl E$.  Since $\omega < 1$, we must have
    $\I(D) > \I(E) = \I(\mathcal{B})$. By induction hypothesis,
    $\I'(D) \ge \I(D)$ and $\I'(E) = \I(E)$, so $\I'(\mathcal{B}) = \I(\mathcal{B})$.
    If $\I(\mathcal{B})>\omega$, then $\I(D)\ge\I(E)=\I(\mathcal{B})>\omega$, by induction
    hypothesis $\I'(D)=\I'(E)=\I'(\mathcal{B})=1$.\\
    $\mathcal{B} \equiv \qe xD(x)$. First assume that $\I(\mathcal{B})\le \omega$. Since $D(c)$
    evaluates to a value less or equal to $\omega$ in $\I$ and, by
    induction hypothesis, in $\I'$ also the supremum of these values
    is less or equal to $\omega$ in $\I'$, thus $\I'(\mathcal{B}) = \I(\mathcal{B})$.
    If $\I(\mathcal{B})>\omega$, then there is a $c$ such that $\I(D(c))>\omega$, by
    induction hypothesis $\I'(D(c))=1$, thus, $\I'(\mathcal{B})=1$.    
\end{proof}
The gluing argument presented above has an important limitation: it does not apply to universal quantifiers if the gluing point \( \omega \) is not isolated from above.
More generally, the gluing lemma holds for all formulas if and only if \( \omega \) is isolated from above in the truth value set.

As a consequence of the gluing lemma we have the following results:
\begin{theorem}\label{p - 2.1} 
For the following cases, 1-satisfiability in Gödel logics coincides with classical satisfiability:\\
1) the propositional case,\\ 
2) the first-order case where the truth value set is arbitrary, provided that 0 is isolated,\\
3) the prenex fragment, for any truth value set, \\
4) the existential fragment, for any truth value set.
\end{theorem}
\begin{proof}
The first two cases follow by induction on the complexity of the formulas.
3) The change of interpretation might change the value for the universal quantifier, but in the modified interpretations it becomes more true.
4) It follows from the fact that purely existential formulas are always classically satisfiable.
\end{proof}
\begin{lemma}
    The gluing lemma holds for all finite Gödel logics.
\end{lemma}
\begin{proof}
    It is obvious since in finite-valued Gödel logics the points are isolated and there is no infinite descending chain of truth values.
\end{proof}
We know that the gluing lemma works whenever there is no universal quantifier or the gluing point is isolated, which is the case for $G_\uparrow$.
\begin{proposition}\label{p - 2.2}\cite{Baaz-Leitsch-Zach:FirstOrderGoedel}
\(G_\uparrow =\bigcap_{m\geq 2} G_m\) is an intersection of all
finite valued Gödel logics, as a consequence the gluing lemma holds for $G_\uparrow$.
\end{proposition}
\begin{proof}
The idea of the proof is the following: a formula valid in $G_\uparrow$ is valid in all finite-valued Gödel logics. On the other hand, the formula not valid in $G_\uparrow$ is not valid in at least one finite-valued Gödel logic. Because the counterexample has a certain finite value, we glue everything from above this value to $1$, then we get a counterexample in finite-valued Gödel logic.
\end{proof}


\begin{remark}\label{r - 2.1}
As a consequence, we observe that $1$-unsatisfiability coincides across all Gödel logics in which $0$ is an isolated point. Moreover, the problem of refuting $1$-satisfiability is decidable, since it reduces to classical logic, where this notion is recursive. 
The same applies to the existential and prenex fragments: refutation of $1$-satisfiability remains recursive, and the relevant Gödel logics agree on this notion whenever $0$ is isolated. 
This highlights a key contrast: while $1$-satisfiability is effectively refutable in these cases, the validity of the negation of a formula (i.e., $\neg A$ being valid) is, in general, \emph{not} recursive.
\end{remark}



\begin{definition}[BS class w.r.t. 1-satisfiability]
A first-order formula $\varphi$ is said to be in the
\emph{Bernays--Sch\"onfinkel class}
if it is of the form
$\exists \bar{x}\,\forall \bar{y}\; \psi(\bar{x},\bar{y})$,
where:
\begin{enumerate}
    \item $\psi$ is a quantifier-free formula,
    \item $\bar{x}$ is a (possibly empty) sequence of existentially quantified variables,
    \item $\bar{y}$ is a (possibly empty) sequence of universally quantified variables,
    \item the signature contains \emph{no function symbols of positive arity}
    (i.e. only constant symbols and predicate symbols).
\end{enumerate}
We refer to $\exists^\ast \forall^\ast$-prefix formulas over a function-free vocabulary as the $BS$ class.
\end{definition}

%
As an immediate consequence of the gluing lemma, we have:
\begin{theorem}
For prenex classes w.r.t. 1-satisfiability the following holds:
    \begin{itemize}
        \item It is decidable in all Gödel logics iff it is classically decidable.
        \item 1-satisfiable sentences of these prenex classes coincide for all Gödel logics.
    \end{itemize}
\end{theorem}

\begin{theorem}\label{theorem1}
1-satisfiability in  BS class is decidable for all G\"odel logics.
\end{theorem}

\begin{proof}
The proof is obvious as $1$-satisfiability coincides with classical satisfiability and, therefore, is decidable.
\end{proof}

 \begin{corollary}
 1-satisfiability of monadic fragments is always decidable if 0 is isolated.
\end{corollary}

\section{Skolemization}
\begin{definition}[weak and strong quantifiers]
Let $A$ be a formula. A strong quantifier in $A$ is a universal quantifier which occurs in a positive position in $A$, or an existential quantifier which occurs in a negative position in $A$. it is dual for weak quantifiers.
\end{definition}
Intuitively, strong quantifiers are those whose witnesses must be fixed in advance
(independent of further context), while weak quantifiers allow variation depending on the context.
This distinction is purely syntactic but corresponds exactly to what the Skolemization
procedures require: strong quantifiers introduce Skolem functions, while weak quantifiers represent context parameters.
This is precisely why the Skolem term for a strong quantifier may contain all weak variables in
its scope, and \emph{only} those.
\begin{definition}[Skolemization]
    Given a formula $A$, its \emph{validity–Skolemization} $A^S$ is obtained by:
\begin{itemize}
  \item replacing every strong existential quantifier
    \(
    \exists x\, B(x)
    \)
    by
    \(
    B(f(\bar y))
    \),
    where $\bar y$ is the list of weak universals in its scope;
  \item replacing every strong universal quantifier
    \(
    \forall x\, B(x)
    \)
    by
    \(
    B(g(\bar y)),
    \)
    where $\bar y$ is the list of weak existentials in its scope.
\end{itemize}
The \emph{1–satisfiability} Skolemization proceeds identically but dual (replacing weak quantifiers).
\end{definition}
\begin{example} 
Consider the formula $\forall x \exists y\forall z A(x,y,z)$. Its validity Skolemization is $\exists y A(c,y,f(c))$, and its $1$-satisfiability Skolemization is $\forall x \forall z A(x,g(x),z)$.
\end{example}
Semantic Skolemization replaces strong quantifiers exactly as above and preserves the intended
semantic property (validity or 1-satisfiability).  
\emph{Syntactic} Skolemization, in contrast, applies the classical rules
\[
\forall x\,B(x) \leadsto B(c) , \qquad
\exists x\,B(x) \leadsto B(f(\bar y)),
\]
without respecting polarity.  
In Gödel logics this may change the truth value of formulas.

\begin{example} 
 \[\left( (\forall x A(x)\supset B) \supset \exists x (A(x)\supset B) \right)^S =  \]
\[( A(c)\supset B) \supset \exists x (A(x)\supset B)  \]
The Skolemization is only valid if the infimum is a minimum in every interpretation.
 \[\left( ( A\supset  \exists xB(x)) \supset \exists x (A\supset B(x)) \right)^S =  \]
  \[( A\supset B(c)) \supset \exists x (A\supset B(x))  \]
  The Skolemization in this case is only valid if the supremum is a maximum in every interpretation.  
\end{example}
Unlike in classical logic, the infimum may not correspond to a minimum, there might be no object in the domain that achieves the exact value of the infimum. This failure to realize the infimum is the core reason Skolemization typically breaks down. Soundness and completeness for the Skolemization of the intermediate logics is given in \cite{QFS}.
  Consequently, the only Gödel logics that do Skolemize are $G_\uparrow$ and finite-valued.
However, for prenex fragments, this issue can be circumvented. If one accepts a version of the axiom of choice, one can select an object whose value is arbitrarily close to the infimum. This selection yields a definable Skolem function, which, in the case of existential quantifiers following universal ones, becomes a function of the universally quantified variables.

In G\"odel logics, valid prenex formulas can be sharpened to validity equivalent purely existential formulas by Skolemization.


\begin{lemma}[Validity Skolemization]\label{lemma1} 
    For all prenex formulas $ Q\bar{x} A(\bar{x})$ and all G\"odel logics $G$
    \begin{align*}
       \Gamma \Vdash_G Q\bar{x} A(\bar{x})
\Longleftrightarrow \Gamma \Vdash_G (Q\bar{x} A(\bar{x}))^S
    \end{align*}
    where $ Q\bar{x} $ is a quantifier prefix and $A(\bar{x})$ is a quantifier-free formula.
\end{lemma}
\begin{proof}
It is sufficient to prove with $A$ arbitrary and $f$ a new function symbol:
    $$\Gamma \Vdash_G  \exists \overline{x} \forall y A(\overline{x},y)\Leftrightarrow\Gamma \Vdash_G  \exists \overline{x} A(\overline{x},f(\overline{x})).$$
The result follows by induction. $(\Rightarrow)$ The direction from left to right is obvious.\\
$(\Leftarrow)$ For the other direction, recall that in Gödel semantics, existential and universal quantifiers are evaluated as the supremum and infimum of the truth values of their instances.

If $\nVdash_G \exists \overline{x} \forall y A(\overline{x},y)$ then for some interpretation $\mathcal{I}$ 
\begin{align*}
    \sup\{d_{\overline{c}}\;|\; \mathcal{I}(\forall y A(\overline{c},y))=d_{\overline{c}}\}\leq d <1.
\end{align*}

Now, suppose that a universally quantified formula $\forall y\, A(c, y)$ has an interpretation value $d_c < 1$. This means that the infimum is strictly below $1$, and therefore the formula is not valid.

The crucial point is how the infimum is handled. We have to consider two cases:
\begin{enumerate}
    \item If the infimum is a minimum, the witness is exact i.e., attained by some specific $y_0$ for a given $x$, then we can directly define a Skolem function $f(x) = y_0$. In this case, the Skolemized formula
$\exists x\, A(x, f(x))$
is also not valid, since it evaluates to $d < 1$.
    \item However, when the infimum is not a minimum, we must use an approximation strategy: Using the axiom of choice, we can select a witness $f(c)$ such that $A(c, f(c))$ evaluates to a value still strictly below $1$, but arbitrarily close to the infimum. 
    \begin{itemize}
    \item We select a value $\varepsilon > 0$ (e.g. $\frac{1-d}{2}$) such that $d_{\overline{c}} + \varepsilon < 1$ (since $d_c$ was the greatest lower bound).
      \item Using the axiom of choice we can assign a value for every $f(\overline{c})$ such that $\mathcal{I}(A(\overline{c},f(\overline{c}))$ is in between $d_{\overline{c}}$ and $\varepsilon$ and the existence is guaranteed due to the density and linearity of the truth value set in Gödel logics.    
   \end{itemize}
   This process is repeated inductively for each quantifier, maintaining the falsity of the formula through approximation. The value increases slightly but stays below $1$, preserving non-validity.
\end{enumerate}

As a consequence 
\begin{align*}
    \sup\{d_{\overline{c}}+\frac{1-d}{2}\;|\; \mathcal{I}(A(\overline{c},f(\overline{c})))\leq \\
    d_{\overline{c}}+\frac{1-d}{2}\}\leq d+\frac{1-d}{2}< 1
\end{align*}
and thus $\Gamma \nVdash_G \exists \overline{x} A(\overline{x},f(\overline{x}))$. 
\end{proof}

\begin{remark}
    This is the key step in semantic Skolemization: it ensures that the existential formula $\exists x\, A(x, f(x))$ remains invalid in any model where the original formula $\exists x\, \forall y\, A(x, y)$ is invalid.
In essence, the construction is two-directional:
\begin{itemize}
    \item If the original formula is valid, so is its Skolemized form.
    \item If the original formula is \emph{not} valid, we can select witnesses to keep the value below $1$ and construct a countermodel for the Skolemized form.
\end{itemize}

\end{remark}


In contrast to validity-Skolemization, 1-satisfiability-Skolemization is straightforward by the gluing argument.


\begin{proposition}
    A prenex formula is 1-satisfiable if its Skolemization is 1-satisfiable.
\end{proposition}

\begin{proof}
    A prenex formula is 1-satisfiable if it is classical. A formula is classical satisfiable if its Skolemization is classical satisfiable.
\end{proof}

\section{Validity in the BS class}
\begin{definition}[BS class w.r.t validity] 
A first-order formula $\varphi$ is said to be in the
\emph{Bernays--Sch\"onfinkel class}
if it is of the form
$\forall\bar{x}\, \exists\bar{y}\; \psi(\bar{x},\bar{y})$,
where:
\begin{enumerate}
    \item $\psi$ is a quantifier-free formula,
    \item $\bar{x}$ is a (possibly empty) sequence of universally  quantified variables,
    \item $\bar{y}$ is a (possibly empty) sequence of existentially quantified variables,
    \item the signature contains \emph{no function symbols of positive arity}
    (i.e. only constant symbols and predicate symbols).
\end{enumerate}
We refer to $\forall^\ast \exists^\ast $-prefix formulas over a function-free vocabulary as the $BS$ class.
\end{definition}
In case of validity Skolemization in the $BS$ class all Skolem functions replacing the universal quantifiers are constants.
\begin{theorem}\label{theorem1}
Validity in the $BS$ class is decidable for all G\"odel logics.
\end{theorem}
\begin{proof}
    From Lemma \ref{lemma1}, it follows that
    \begin{align*}
       \Gamma \Vdash_G \forall \bar{x} \exists \bar{y}A(\bar{x},\bar{y})
\Longleftrightarrow \Gamma \Vdash_G  \exists \bar{y}A(\bar{c},\bar{y})
    \end{align*}
for new constants $\bar{c}$.
Suppose there is a countermodel $M$ such that
$M \nVdash_G  \exists \bar{y}A(\bar{c},\bar{y})$.
Then there is also a countermodel $M'$ such that $M' \nVdash_G  \exists \bar{y}A(\bar{c},\bar{y})$ where the domain of $M'$ contains only interpretations of $\bar{c}$.
\end{proof}
In most Gödel logics, especially those that are not r.e. or lack an analytic calculus, no Herbrand disjunction or syntactic decision procedure exists. However, for the $BS$ class, Skolemization leads directly to a finite disjunction of quantifier-free ground instances:
\[
A(\bar{c_1}, \bar{d_1}) \vee A(\bar{c_2}, \bar{d_2}) \vee \dots \vee A(\bar{c_n}, \bar{d_n})
\]
This finite Herbrand disjunction can be evaluated propositionally, ensuring decidability. All Gödel fragments whose quantifier alternation can be reduced to finite Herbrand disjunctions are decidable.
\begin{corollary}\label{cor}
\begin{enumerate}
\item[1)]  Let $\exists \bar{y}A(\bar{y})$
contain only constants $\bar{c}$, then Herbrand's theorem holds for $\exists \bar{y}A(\bar{y})$ for all G\"odel logics $G$.
\item[2)] Let $\forall \bar{x} \exists \bar{y}A(\bar{x},\bar{y})$ prenex formulas contain only constants $\bar{d}$, then  $\Gamma \Vdash_G \forall \bar{x} \exists \bar{y}A(\bar{x},\bar{y})
\Longleftrightarrow \Gamma \Vdash_{G'} \forall \bar{x} \exists \bar{y}A(\bar{x},\bar{y}) $
for all infinitely-valued G\"odel logics $G$, $G'$.
    \end{enumerate}
\end{corollary}

\begin{proof}
1) According to the theorem \ref{theorem1},
$M \nVdash_G  \exists \bar{y}A(\bar{c},\bar{y})$ implies $M' \nVdash_G  \exists \bar{y}A(\bar{c},\bar{y})$ with restricted domain to constants only.\\
    2) follows from 1), as the Herbrand disjunction is contained in $\bigvee_nA(\bar{c}_n, \bar{d}_n)$ where $\bar{c}_n, \bar{d}_n$ are possible variations of $\bar{c}, \bar{d}$ and validity for propositional formulas coincides with infinitely-valued G\"odel logics.
\end{proof}
\begin{remark}
  Note that  1) is not trivial as prenex formulas and consequently $\exists$-formulas (see. Lemma \ref{lemma1}) for countable G\"odel logics are not r.e. \cite{Baaz-Preining-Zach:GoedelFragments}.
\end{remark}


\begin{corollary}
All Gödel logics coincide for
the $BS$ class w.r.t. 1-satisfiability, but only the infinite-valued Gödel logics coincide for the $BS$ class w.r.t. validity. 
The finite-valued Gödel logics have less and less valid formulas in the $BS$ class relative to their increasing truth values. 
\end{corollary}
\begin{proof}
The coincidence in the infinite case is given by coincidence in the propositional formulas and relates to the same Herbrands disjunction as shown in Proposition \ref{p - 1.2}. The difference is induced by the difference on the propositional part. 
\end{proof}
Negative classical results (undecidability) transfer to Gödel logics because the validity of $BS$ formulas coincides across classical and Gödel semantics.
\begin{corollary}
The BS fragment of any infinitely-valued G\"odel logic is the intersection of the $BS$ fragments of the finite-valued G\"odel logics, both for 1-satisfiability and validity
\[
\text{BS}({G_{[0,1]}}) = \bigcap \text{BS}({G_n}).
\]
\end{corollary}

\begin{proof}

The truth of $BS$ formulas depends only on the Herbrand disjunction.
Each Herbrand disjunction involves finitely many atoms; Once the Gödel set has more truth values than these atoms, further refinements do not affect validity.
Therefore, if a $BS$ formula holds in the infinite Gödel logic, it holds in all sufficiently large finite Gödel logics, and vice versa.
\end{proof}

\begin{remark}
The $BS$ class of any infinite Gödel logic is a r.e. sub-theory of the full (sometimes non–r.e.) first-order Gödel logic. As an example, consider \[
\text{BS}({G_{[0,1]}}) \subseteq \text{BS}({G_n})\subseteq {G_n}
\] and similarly for all finite Gödel logics. Hence, decidability of the $BS$ class follows constructively from the Skolemization lemma \ref{lemma1}. 
\end{remark}

\section{Conclusion and Further Extensions}
As 1-satisfiability is not dual to validity as in classic logic, other notions of satisfiability can be formulated , e.g., $>0$-satisfiability:

\begin{definition}
A formula in Gödel logic is \emph{$>0$-satisfiable
} if there exists at least one interpretation
that assigns $>0$
to the formula. $>0$-satisfiability of a formula $A$ coincides with 1-satisfiability of $\neg \neg A$. 
\end{definition}
As the gluing lemma does not hold for prenex formuals in $>0$- satisfiability, the status of the $BS$ class remains open. A natural direction for extension is the study of \emph{other} subclasses of prenex form as the Ackermann class, Gödel class, monadic fragments, or guarded fragments. Here the results might depend on the existencee of the Herbrand disjunction for prenex formulas which hold for finite-valued Gödel logic, standard $G_{[0,1]}$ but not for $G_\uparrow$. Countable Gödel logics do not have Herbrand disjunction \cite{Baaz-Zach:Compact}.



\end{document}